\documentclass{l4dc2025}

\title[Predictive CBFs]{Learning for Layered Safety-Critical Control with \\ Predictive Control Barrier Functions}
\usepackage{times}
\usepackage{algorithm}
\usepackage{algpseudocode}
\SetAlgoLined
\usepackage{style}

\newtheorem*{untheorem}{Theorem}

\coltauthor{\Name{William D. Compton} \Email{wcompton@caltech.edu}\\
 \Name{Max H. Cohen} \Email{maxcohen@caltech.edu}\\
 \Name{Aaron D. Ames} \Email{ames@caltech.edu}\\
 \addr Caltech, Pasadena, CA, USA}

\begin{document}

\maketitle

\vspace{-0.5cm}
\begin{abstract}%
Safety filters leveraging control barrier functions (CBFs) are highly effective for enforcing safe behavior on complex systems.  It is often easier to synthesize CBFs for a Reduced order Model (RoM), and track the resulting safe behavior on the Full order Model (FoM)---yet gaps between the RoM and FoM can result in safety violations. This paper introduces \emph{predictive CBFs} to address this gap by leveraging rollouts of the FoM to define a predictive robustness term added to the RoM CBF condition.  Theoretically, we prove that this guarantees safety in a layered control implementation.  Practically, we learn the predictive robustness term through massive parallel simulation with domain randomization.  We demonstrate in simulation that this yields safe FoM behavior with minimal conservatism, and experimentally realize predictive CBFs on a 3D hopping robot. 

\end{abstract}


\section{Introduction}

Safety is a primary concern in the construction of modern control architectures. Control Barrier Functions (CBFs) provide a theoretically grounded method for achieving safe behaviors of complex nonlinear systems \citep{ames2014control}. CBFs have several desirable properties, including nominal robustness properties \citep{xu2015robustness} and the ability to filter desired system inputs for safety (safety filtering) \citep{ames2016control,ames2019control}. These benefits have established CBFs and safety filters as practical tools, commonly used in collision avoidance, and other applications, on various robotic systems \citep{wabersich2023data}. 
Yet synthesizing CBFs for arbitrary (often even simple) safety constraints is difficult and impractical for high-dimensional or complex nonlinear systems. 
Significant prior work exists surrounding constructive methods for CBF synthesis, leveraging different system properties, such as relative degree \citep{nguyen2016exponential}, differential flatness \citep{wang2017safe}, HJ reachability \citep{tonkens2022refining}, backup controllers \citep{chen2021backup}, or learning \citep{robey2020learning}.
However, complex systems often will not satisfy required assumptions, and high dimensionality cripples learning or PDE-based methods. 
Hence, constructive synthesis of CBFs for complex, high-order systems is an open problem \citep{cohen2024safety}. 

One recent approach to addressing this problem is through layered control architectures \citep{matni2024quantitative}, which leverage reduced order models in conjunction with existing controllers lower in the control stack.
Often, safety constraints depend on a subset of system states; for instance, collisional avoidance problems typically only depend on kinematic configuration, and safety requirements for mobile robots often concern their center of mass positions.
A reduced order model (RoM) approximates the dynamics of the full order model (FoM) while capturing all of the states relevant for safety.
CBF synthesis can then be carried out on the RoM, and resulting behaviors can be tracked by a tracking controller on the FoM.
Under various assumptions regarding the tracking controller, e.g., exponential tracking \citep{Molnar2022, singletary2022safety} or the existence of a simulation function \citep{cohen2025safety}, safety can be extended to the FoM from a reduced set of initial conditions. 
Yet certifying the assumptions on the tracking controller is, in itself, a hard problem---and calculating the corresponding certificates is even more challenging. 
In practice, terms in the resulting safety filters are conservatively approximated to achieve safety empirically.

The use of RoMs to construct CBFs for safety parallels the planner-tracker paradigm, where trajectories are optimized over a RoM and then tracked by a tracking controller on the FoM.
Methods in this framework focus primarily on establishing tracking error bounds, through HJ Reachability \citep{mitchell2005time}, sum of squares programming \citep{schweidel2022safe}, or contraction theory \citep{singh2023robust}.
Once a tracking error bound has been established, tools from tube MPC \citep{langson2004robust, compton2025dynamic} or search-based planning \citep{chen2021fastrack} can be used to construct safe trajectories of the FoM.
As these certification problems scale poorly with system complexity and dimensionality, we focus instead on leveraging the simplicity of the layered CBF paradigm with RoMs coupled with the power of a predictive horizon to achieve safety on the FoM.

\begin{figure}
    \centering
    \includegraphics[width=\textwidth]{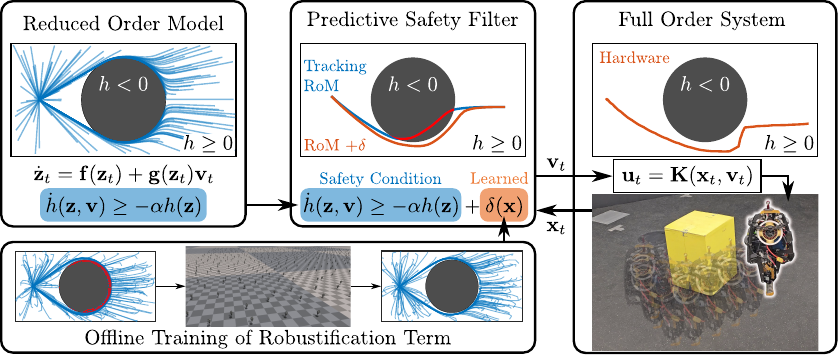}
    \caption{\vspace{-0.9cm}First, a CBF is synthesized on a RoM; tracking errors from the FoM can lead to unsafe behavior. We learn a structured robustification term $\delta(\b x)$ from massively parallel simulation, and deploy the robustified CBF on hardware to obtain safe behaviors.\vspace{-1.0cm}}
    \label{fig:hero}
\end{figure}

\subsection{Our Contribution: }
Rather than relying upon restrictive and difficult-to-verify tracking assumptions that induce conservative behavior, we learn a robustification of the CBF synthesized on the RoM to remove safety violations over a horizon-based rollout. 
This robustification term, guaranteed to exist under mild assumptions on tracking error, is conditioned on the FoM state. 
When added to the RoM safety filter, it induces a minimal constraint back-off to ensure safety of the FoM.  We both formally establish guarantees of safety, and leverage massively parallel simulation to learn this term with high accuracy.  
We demonstrate our method outperforming previous methods using RoMs for layered safety filtering by attaining safety under conditions where previous methods fail (due to violation of their underlying assumptions), or expanding the safe set when previous methods were successful.
Finally, we demonstrate our method on a complex robotic system, a 3D hopping robot ARCHER, achieving safe navigation of cluttered environments on hardware.

\newpage 

\section{Background}

We begin by defining a framework for layered control systems consisting of a reduced order model (RoM) and full order model (FoM), interacting through a safety filter (as illustrated in Fig. \ref{fig:architecture}).

\begin{wrapfigure}[19]{L}{0.52\textwidth} 
  \vspace{-0.5cm} 
  \begin{center}
  \begin{scriptsize}
  \begin{tikzpicture}[
  mynode/.style = {draw, text width=0.415\textwidth, align=center, minimum height=0.8cm}]
  
  \node[mynode] (high) {\textbf{Reduced order Model (RoM)}
  \vspace{-0.3cm}
    \begin{align*}
    \textrm{Dynamics:} \quad & \dot{\b z}_t = \b f(\b z_t) + \b g(\b z_t) \b v_t, \\ 
    \textrm{Safety:} \quad & h(\b z_t) \geq 0, \quad \forall t \geq 0.
    \end{align*}
  };
    
  \node[mynode, below=of high] (mid) {\textbf{Predictive CBF}
  \vspace{-0.3cm}
    \begin{align*}
    \b k_{sf}^{\delta}(\b z, \b x) = & \argmin _{\b v \in \cal V} \quad \frac{1}{2}\| \b v - \b k(\b z) \|^2 \\
    \mathrm{s.t.} \quad  & \underbrace{\L_f h(\b z)  + \L_g h(\b z)  \b v \geq -\alpha h(\b z)}_{\textrm{RoM CBF Condition}}+ \hspace{-0.1cm} \underbrace{\delta(\b x)}_{\textrm{Learned}}
    \end{align*}
  };

  \node[mynode, below=of mid] (low) {\textbf{Full order Model (FoM)}
  \vspace{-0.3cm}
    \begin{align*}
    \textrm{Dynamics:} \quad & \dot{\b x}_t = \bF (\bx_t) + \bG(\bx_t) \b K(\b x_t,\b v_t), \\ 
    \textrm{Safety:} \quad & h(\Pi(\b x_t)) \geq 0, \quad \forall t \geq 0.
    \end{align*}
  };

  \begin{scope}[transform canvas={xshift=.6em}]
    \draw[latex-] (high) --  node[right] {$\begin{array}{l} \textbf{Top Interface} \\ \b v_t = \b k_{sf}^{\delta}(\b z_t, \b x_t) \end{array}$} ++(mid);
  \end{scope}
  \begin{scope}[transform canvas={xshift=-.6em}]
    \draw[-latex] (high) --  node[left] {$\b z_t $} ++(mid);
  \end{scope}

    \begin{scope}[transform canvas={xshift=.6em}]
        \draw[-latex] (mid) --  node[right] {$ 
    \begin{array}{l} \textbf{Bottom Interface} \\ 
\b v_t = \b k_{sf}^{\delta}(\bPi(\bx_t), \b x_t)
\end{array}
$  }  ++  (low);
    \end{scope}
    \begin{scope}[transform canvas={xshift=-.6em}]
    \draw[latex-] (mid) --  node[left] {$\b x_t $ }  ++  (low);
    \end{scope}

  \draw [latex-latex] (low.east) -- ++(0.8cm,0) |- (high) node[pos=0.25,sloped,above]{State Projection: $z = \Pi(x)$};
  \draw [latex-latex] (low.east) -- ++(0.8cm,0) |- (high) node[pos=0.25,sloped,below]{Input Projection: $v = \Psi(x)$};
  
  \end{tikzpicture}
  \vspace{-.4cm}
  \end{scriptsize}
  \end{center}
  \caption{Layered architecture of a RoM and FoM interacting through a Predictive CBF.}
  \label{fig:architecture}
\end{wrapfigure}
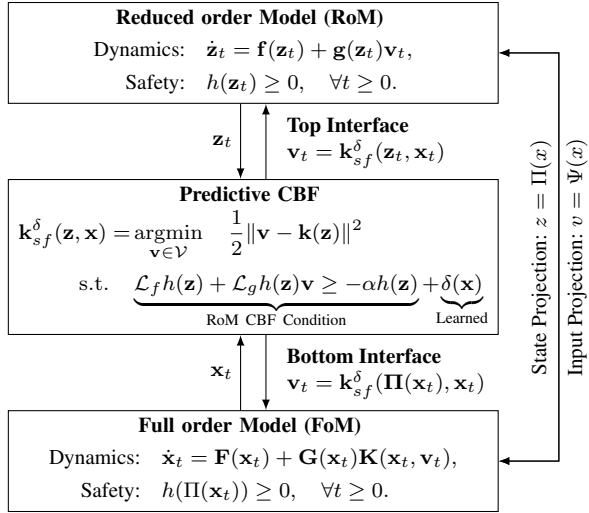
\newsec{Full Order Model (FoM):}  Consider a full order system representing the system of interest---typically these are very complex, high dimensional, and often not known.  We assume that the full-order system evolves according to the nonlinear control system: 
$$
\dot{\b x}_t =\bF (\bx_t) + \bG(\bx_t) \bu_t
$$
where $\bF: \cal X  \to \R^N$ and $\bG : \cal X \to \R^N \times \R^M$ are possibly \emph{unknown} (but can be simulated), and the state $\bx \in \cal X \subseteq \R^N$ and input $\bu \in \cal U \subseteq \R^M$ typically have high dimension.

\newsec{Reduced Order Model (RoM):}  Assume there is access to a reduced order model:
\begin{eqnarray}
\label{eqn:RoM}
\dot{\b z}_t = \b f(\b z_t) + \b g(\b z_t) \b v_t
\end{eqnarray}
with state $\b z \in \cal Z \subseteq \R^n$, input $\b v \in \cal V \subset \R^m$.

\newsec{Coupling between RoM and FOM:}  
The RoM is related to the FoM through a tracking controller---signals are sent from the RoM to the FoM and tracked by an existing (low-level) controller: $\b u = \b K(\b x,\b v)$.  The result is the partially closed-loop full-order dynamics: 
\begin{eqnarray}
\label{eqn:FoMcl}
\dot{\b x}_t = \b F_{\mathrm{cl}} (\b x_t, \b v_t) := 
\bF (\bx_t) + \bG(\bx_t) \bK(\b x_t,\b v_t)
\end{eqnarray}
The RoM is related to the FoM through a projection map: $\b \Pi: \cal X \to \cal Z$. Therefore, given a feedback controller for the RoM, $\b v = \b k(\b z)$, the result is the closed loop dynamics on the FoM: 
\begin{eqnarray}
\label{eqn:FoMclcl}
\dot{\bx}_t = \bF_{\mathrm{cl}}(\bx_t, \bk(\bPi(\bx_t))) = 
\bF (\bx_t) + \bG(\bx_t) \bK(\bx_t,\bk(\bPi(\bx_t)))
\end{eqnarray}
We are interested in the behavior of this closed-loop system, especially in the context of safety. 

\newsec{Safety:}  Consider a safety constraint for the RoM encoded by a smooth function $h : \cal Z  \to \R$: 
\begin{eqnarray}
\C_{\mathrm{RoM}} = \{ \b z \in \cal Z \subseteq \R^n ~ : ~ h(\b z) \geq 0\}, 
\end{eqnarray}
with $\partial\C_{\mathrm{RoM}} = h^{-1}(0)$. For the FoM, this safety constraint takes the form: 
\begin{eqnarray}
\C_{\mathrm{FoM}} = \{ \b x \in \cal X \subseteq \R^N ~ : ~ h( \b \Pi(\b x)) \geq 0 \}. 
\end{eqnarray}
The goal is to enforce safety on the FoM through a controller applied to the RoM.  That is: 

\begin{problem}\label{prob:main}
    Given a RoM \eqref{eqn:RoM}, a closed-loop FoM \eqref{eqn:FoMcl} and a safe set $\C_{\mathrm{FoM}}$, find a controller for the RoM $\b v = \b k(\b z)$ and a set $\mc{S} \subseteq \C_{\mathrm{FoM}}$ such that $\bx_0 \in\mc{S}$ implies that $\bx_t \in C_{\mathrm{FoM}}$ for all $t \geq 0$. 
\end{problem}

\newsec{Safety on RoMs with CBFs:}  There are well-established results on using RoMs to guarantee safety on the FoM.  Yet these only guarantee safety under the assumption of exponential tracking \textit{without} steady-state error---which never holds in practice. 
The goal of this paper is to relax this conservatism using predictive CBFs.  First, we summarize the key result for safety on FoMs using RoMs. 

Assume that $h$ is a \emph{Control Barrier Function (CBF)} \citep{ames2016control} for the RoM \eqref{eqn:RoM}:
\begin{eqnarray}
\label{eqn:CBF}
\sup_{\b v \in \cal V} \; \Big[\dot{h}(\b z, \b v) = \underbrace{\pdv{h}{\bz} \Big \vert_{\bz}  \b f(\b z)}_{:=\L_f h(\b z)} + \underbrace{\pdv{h}{\bz} \Big \vert_{\bz}  \b g(\b z)}_{:=\L_g h(\b z)} \b v \Big] \geq - \alpha(h(\b z)) ,
\end{eqnarray}
for all $\b z \in \cal Z$, where $\alpha$ is an extended class $\cal K$ function.  The result is a \emph{safety filter} for the RoM: 
\begin{align}
\label{eqn:safetyfilter}
\b k_{\mathrm{sf}}(\b z) = \argmin _{\b v \in \cal V} &  \quad \frac{1}{2}\| \b v - \b k(\b z) \|^2 \\
\mathrm{s.t.} \quad  &  \quad \L_f h(\b z)  + \L_g h(\b z)  \b v \geq -\alpha h(\b z)  \nonumber
\end{align}
where for simplicity we take $\alpha(h(\b z)) = \alpha h(\b z)$ for $\alpha \in \R_{\geq 0}$. 

This can guarantee safety for the FoM under the proper set of assumptions on the RoM and the interface between the RoM and FoM.  Namely, ``relative degree'' and ``boundedness'' assumptions:


\begin{assumption}
\label{asm:projection}
There is a projection $\bPsi : \cal X \to \cal V$ from the FoM state to the RoM inputs such that: 
\begin{equation}
    \pdv{\bPi}{\bx} \Big \vert_{\bx} \bF(\bx) = \bf(\bPi(\bx)) + \bg(\bPi(\bx))\bPsi(\bx), \qquad 
    \pdv{\bPi}{\bx} \Big \vert_{\bx} \bG(\bx)\equiv \bzero.
\end{equation}
\end{assumption}

\begin{assumption}
\label{asm:bounded}
    $\| \L_g h(\bPi(\bx)) \| \leq C_h$, for $C_h \in\R_{> 0}$ and all $\bx \in \mathcal C_{\mathrm{FoM}}$. 
\end{assumption}

Assumption \ref{asm:projection} and \ref{asm:bounded} are satisfied for all of the layered systems considered in this paper. 

\begin{remark}
Assumption \ref{asm:projection} is often satisfied by robotic systems.  In particular, in this setting, the FoM is described by: $\bD(\bq) \ddot{\bq} + \bH(\bq,\dot{\bq}) = \bB \bu$ where $\bq \in \Q \subset \R^n$ is the configuration variables (positions and angles) and $\dot{\bq} \in T_{\bq}\Q $ are the velocities and $\bu \in \R^M$ the control inputs.  The corresponding FoM state is therefore: $\bx = (\bq,\dot{\bq}) \in \X \subset T \Q \subset \R^{2n}$.  In this setting RoM are often kinematic---the most prevalent of which is the single integrator: $\dot{\bz} = \dot{\bq} = \bv$.  It follows that $\cal Z \subset \R^n$ and $\cal V \subset \R^n$.  The projections are therefore $\bPi: \X \to \cal Z$ and $\bPsi : \X \to \cal V$ given in coordinates by: $ \bPi(\bq,\dot{\bq}) = \bq$ and $ \bPsi(\bq,\dot{\bq}) = \dot{\bq}$.  It follows that $\pdv{\bPi}{\bx} \Big \vert_{\bx} \bF(\bx) = \dot{\bq} = \dot{\bz}$ and $\pdv{\bPi}{\bx}(\bx)\bG(\bx)\equiv \bzero$ as desired. 
\end{remark}

Finally, assume that the reference signal produced by the safety filter, $\bv = \b k_{\mathrm{sf}}(\b z)$ can be tracked exponentially by the FoM.  This is encoded by a positive definite \textbf{tracking function} $V : \X \to \R_{\geq 0}$: 
\begin{eqnarray}
\label{eqn:trackingfunction}
\rho \| \bPsi(\bx) -  \bk_{\mathrm{sf}}(\bPi(\bx)) \| & \leq & V(\bx),\qquad \rho\in\R_{>0}, \\
 \dot{V}(\bx_t) = \pdv{V}{\bx} \Big \vert_{\bx_t}   \bF_{\mathrm{cl}}(\bx_t, \bk_{\mathrm{sf}}(\bPi(\bx_t))) & \leq & - \lambda V(\bx_t), \qquad \lambda \in \R_{> 0}. \nonumber
\end{eqnarray}
The main result of \citep{Molnar2022} establishes the safety of the FoM through safety filters on the RoM, per Problem \ref{prob:main}, as summarized by the following (stated in the notation of this paper).

\begin{theorem}
\label{eqn:thmbackground}
Consider a RoM \eqref{eqn:RoM}, a closed-loop FoM \eqref{eqn:FoMcl} and a safe set $\C_{\mathrm{RoM}}$ satisfying Assumptions \ref{asm:projection} and \ref{asm:bounded}. 
For the safety filter $\bv = \b k_{\mathrm{sf}}(\b z)$ in \eqref{eqn:safetyfilter}, if there exists a tracking function $V$ satisfying \eqref{eqn:trackingfunction} with $\lambda \geq \alpha_x$ for any $\alpha_x > \alpha$, then the FoM is safe: 
\begin{eqnarray}
\bx_0 \in \cal{S} := \left\{ 
\bx \in \X  \colon (\alpha_x - \alpha) h(\bPi(\bx)) - \frac{C_h}{\rho} V(\bx) \geq 0 \right\}
\quad \implies \quad 
\bx_t \in \C_{\mathrm{FoM}}, \:\: \forall t \geq 0. 
\end{eqnarray}
\end{theorem}

\section{Predictive CBFs}
We now present the main concept of this paper: \emph{Predictive CBFs}.  These account for variation between the RoM and the FoM by finding a constant to buffer the CBF condition.  It does this in a minimally conservative fashion through a rollout step, i.e., through prediction.  We establish two key theoretic properties of Predictive CBFs: they exist, and under reasonable assumptions guarantee safety for both the RoM and the FoM. 
 Predictive CBFs, therefore, solve Problem \ref{prob:main}.

\subsection{Safety with Predictive CBFs}

Predictive CBFs forward evaluate the FoM to find a minimal robustness term, $\delta$, in which to buffer the safety condition for the RoM and thereby guarantee safety for the FoM.

\begin{definition}
  Given a RoM \eqref{eqn:RoM} and a FoM with tracking controller \eqref{eqn:FoMcl}, $h$ is a \textbf{Predictive CBF} if: 
  \begin{itemize}
      \item[(i)] $h$ is a CBF for the RoM \eqref{eqn:CBF}, 
      \item[(ii)] There exists a subset $\cal C_{\mathrm{FoM}}^{\delta} \subseteq \cal C_{\mathrm{FoM}}$ and a \textbf{predictor} $\delta : \mc{C}_{\mathrm{FoM}}^{\delta} \subset \cal X \to \R_{\geq 0}$ satisfying: 
  \begin{align}
\label{eqn:delta_optimization}
\delta(\b x)  = \min _{\delta \in \R_{\geq 0}} &  \quad \delta  & \\
\mathrm{s.t.}    &  
\quad \dot{h}(\b \Pi(\b x_t),\dot{\b  x}_t) \geq -\alpha_{\b x} h(\b \Pi(\b x_t)),    & \alpha_x \in \R_{\geq 0}, ~  \forall t \geq 0 \nonumber\\
 &  
 \quad 
\L_f h(\b z)  + \L_g h(\b z)  \b v^{\delta}(\b z) \geq -\alpha h(\b z) + \delta,   & \b z = \Pi(\b x)  \nonumber\\
 & 
 \quad \dot{\b  x}_t = \b  F_{\mathrm{cl}}(\b x_t, \b K(\b x_t,\b v^{\delta}(\b \Pi(\b x_t))),  &  \b x_0 = \b x .  \nonumber
\end{align}
\end{itemize}
The corresponding \textbf{predictive safety filter} is given by: 
\begin{align}
\label{eqn:predictive_safetyfilter}
\b k_{sf}^{\delta}(\b z, \b x) = \argmin _{\b v \in \cal V} &  \quad \frac{1}{2}\| \b v - \b k(\b z) \|^2 \\
\mathrm{s.t.} \quad  &  \quad \L_f h(\b z)  + \L_g h(\b z)  \b v \geq -\alpha h(\b z) + \delta(\b x). \nonumber
\end{align}
\end{definition}

\begin{remark}
The name ``predictive CBF'' has appeared in the literature in different forms than considered here, notably in \cite{breeden2022predictive} and \cite{wabersich2022predictive}.  In the former, trajectories are rolled out and the impact of the barrier function is considered.  In the latter, discrete time CBFs are considered in a model predictive control setting.  The notion considered here differs in that (1) it is considered for continuous-time systems, (2) the predictive nature of the CBF is in the CBF condition itself via $\delta(\bx)$, and (3) we consider Predictive CBFs in the setting of layered control systems to use prediction to encapsulate the difference, $\delta$, between the RoM and FoM. 

The closest method, philosophically, to Predictive CBFs is Projection to State Safety (PSS) \citep{taylor2020control}, which adds a delta term to the CBF condition to ensure approximate safety---input to state safety (ISSf).  
Yet this relied on learning $\delta$ from past trajectories, while PCBFs determine $\delta$ from future behaviors.  Formally, this guarantees safety rather than approximate safety (ISSf). 

Lastly, note that in situations where tracking assumptions of previous works, i.e. \cite{Molnar2022}, are satisfied, this algorithm returns $\delta=0$, but will potentially expand the FoM safe set.
\end{remark}

Importantly, it follows by definition that predictive CBFs formally guarantee safety on both the reduced order and full order models simultaneously solving Problem \ref{prob:main}.  Namely, the satisfaction of $\dot{h}(\b \Pi(\b x_t),\dot{\b  x}_t) \geq -\alpha_{\b x} h(\b \Pi(\b x_t))$ in \eqref{eqn:delta_optimization} implies $h(\bPi(\bx_t)) \geq 0$, cf. \citep{ames2016control}, yielding: 

\begin{theorem}[Predictive CBFs $\implies$ Safety]
  If $h$ is a predictive CBF, then the FoM is safe:
  \begin{eqnarray}
    \dot{\b x}_t = \b F_{\mathrm{cl}}(\b x_t,\b K(\b x_t,\b k_{sf}^{\delta}(\b  z_t,\b x_t))   & \qquad \implies \qquad  &   
    \textrm{if} ~  \b x_0 \in \cal C_{\mathrm{FoM}}^{\delta} \quad \textrm{then} \quad  \b x_t \in \cal C_{\mathrm{FoM}} \quad \forall t \geq 0. \nonumber
  \end{eqnarray}
\end{theorem}

\subsection{Existence of Predictive CBFs}
If we make a few key assumptions, it is possible to show that Predictive CBFs are well-defined.  These are analogous to the assumptions needed for Theorem \ref{eqn:thmbackground}, except that the addition of the $\delta$ term allows us to relax the tracking assumption to bounded tracking---this is critical in real-world applications where exact tracking is never achieved.  


Consider the reference signal produced by the predictive safety filter, $\bv = \bk_{\mathrm{sf}}^{\delta}(\bz,\bx)$ given in \eqref{eqn:predictive_safetyfilter}.  Assume a positive definite \textbf{$\delta$-bounded tracking function} $V_{\delta} : \X \times \R_{\geq 0} \to \R_{\geq 0}$ satisfying: 
\begin{align}
\label{eqn:deltatrackingfunction}
\rho \| \bPsi(\bx) -  \bk_{\mathrm{sf}}^{\delta}(\bPi(\bx),\bx) \|  ~ \leq ~ &   V_{\delta}(\bx)  & \rho\in\R_{>0} \\
 \dot{V}_{\delta}(\bx_t)  ~  \leq ~  & - \lambda V_{\delta} (\bx_t) + \mu,  & \mu \in \R_{\geq 0}, \:\:  \lambda \in \R_{> 0}. \nonumber
\end{align}
where $\mu$ is a steady state tracking error.   Note that when $\mu = 0$, $V_{\delta(\bx) \equiv 0}$ is a tracking function as in \eqref{eqn:trackingfunction}.  Bounded tracking functions, therefore, capture the more challenging case of tracking only to an error tube (dictated by $\mu$).  It is the addition of $\delta$ in the safety filter \eqref{eqn:predictive_safetyfilter} that guarantees safety.

\begin{theorem}[Existence of Predictive CBFs]\label{thm:predcbfs-exists}
Consider a RoM \eqref{eqn:RoM}, a closed-loop FoM \eqref{eqn:FoMcl} and a safe set $\C_{\mathrm{RoM}}$ satisfying Assumptions \ref{asm:projection} and \ref{asm:bounded}.  Let $\delta_0(\bx) \equiv \delta_0$ for $\delta_0 \in \R_{\geq 0}$. If there exists a $\delta_0$-bounded tracking function, $V_{\delta_0}$, with $\lambda \geq \alpha_x$ for any $\alpha_x > \alpha$, then the FoM is safe: 
\begin{equation}
\begin{array}{rcl}
\delta_0 & \geq & \frac{C_h \mu}{\alpha_x \rho}\\
\bx_0 \in \cal{S}_{\delta_0} & : = & \C_{\mathrm{FoM}} \cap \left\{ 
\bx \in \X  \colon (\alpha_x - \alpha) h(\bPi(\bx)) + \delta_0 - \frac{C_h}{\rho} V_{\delta_0}(\bx) \geq 0 \right\}
\end{array}
\:\: \implies  \:\: 
\begin{array}{r}
\bx_t \in \C_{\mathrm{FoM}} \\
\:\: \forall t \geq 0. 
\end{array} \nonumber
\end{equation}
Moreover, the set $\cal{S}_{\delta_0} \subset \cal C_{\mathrm{FoM}}$ and the constant function $\delta_0:\cal{S}_{\delta_0}  \to \R_{\geq 0}$ satisfies $\delta(\b x) \leq \delta_0$ for all $\b x \in \cal{S}_{\delta_0}$.  Therefore, $\delta : \C_{\mathrm{FoM}}^{\delta} \to \R_{\geq 0}$ exists with $\cal{S}_{\delta_0}  \subset \C_{\mathrm{FoM}}^{\delta}$. 
\end{theorem}


The proof of this theorem, provided in the appendix, constructed an approximation to $\delta(\b x)$ that did not require direct knowledge of the FoM---only the corresponding performance of the tracking controller for that system via $V_{\delta}$.  Determining this is not easy, but can be estimated in practice.   This motivates the idea of directly learning $\delta$. Additionally, it will be seen that learning $\delta$ enables us to ensure safety even when all of the assumptions of the theorem do not hold, i.e., when $\lambda < \alpha_x$. 

\newsec{Simulation Results.}  
To demonstrate improvements of the predictive CBF over previous methods, we will consider a double integrator tracking velocity commands synthesized from single integrator dynamics. 
The RoM, FoM and tracking controller, i.e., the key elements in Fig. \ref{fig:architecture}, are therefore: 
\begin{equation}
\label{eqn:singledoubleint}
    \dot{\b z} = \b v \qquad \dot{\b x} = \begin{bmatrix}
       \b 0 & \b I \\ \b 0 & \b 0
    \end{bmatrix}\b x + \begin{bmatrix}
        \b 0 \\ \b I
    \end{bmatrix} \b u  \qquad 
    \begin{array}{c}
    \b\Pi(\b x) = \begin{bmatrix}
        \b I & \b 0
    \end{bmatrix}\b x \\
    \b\Psi(\b x) = \begin{bmatrix}
      \b 0   &  \b I
    \end{bmatrix}\b x
    \end{array}
    \quad\quad \b K(\b x, \b v) = -k_v\left(\b \Psi(\b x) - \b v\right)
\end{equation}

\vspace{-0.2cm}
\begin{example}[Single/Double Integrator]
\label{example:scalar}
\emph{
To examine properties of $\delta(\b x)$, consider the scalar RoM case of \eqref{eqn:singledoubleint}, i.e., $z \in \R, \b x \in \R^2$, i.e., the RoM is $\dot{z} = v$ and the FoM is $\dot{x}_1 = x_2$ and $\dot{x}_2 = u = - k_v(x_2 - v)$.  The safety constraint is to maintain nonnegative position, $h(z) = z$, i.e., $\C_{\mathrm{RoM}} = \{ z \in \R : h(z) \geq 0\}$ and $\C_{\mathrm{FoM}} = \{ \bx \in \R^2 : h(\bPi(\bx)) = x_1 \geq 0\}$.  Finally, fix the desired controller in the safety filter \eqref{eqn:predictive_safetyfilter} to be $k(z) = -\tfrac{1}{2}$, which attempts to drive the single integrator into the unsafe region. Following Theorem \ref{thm:predcbfs-exists}, we instantiate a predictive safety filter using the constant function $\delta_0(\bx)\equiv\delta_0$ and examine the properties of the closed-loop system for different values of $\delta_0$ (Fig. \ref{fig:ds_1d}, top row). Here, when $\delta_0=0$ (left), the conditions of Theorem \ref{thm:predcbfs-exists} are not satisfied, leading to safety violations even when starting within $\mc{S}_{\delta_0}$. On the other hand, increasing $\delta_0$ to $1$ ensures the conditions of this result are satisfied, leading to safety for initial conditions within $\mc{S}_{\delta_0}$ (middle). 
}
\end{example}

\begin{figure}
    \centering
    \includegraphics[width=\textwidth]{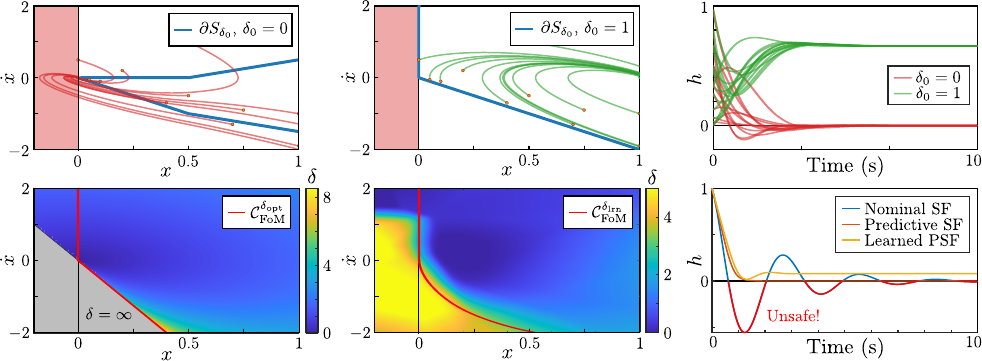}
    \caption{\vspace{-0.5cm}
    \textbf{Top:} Comparison of the set $\mc{S}_{\delta_0}$ from Theorem \ref{thm:predcbfs-exists} for different values of $\delta_0$ and trajectories of the closed-loop FoM in Example \ref{example:scalar}. 
    \textbf{Bottom:} Comparison between the optimized (left) and learned (middle) values of $\delta(\b x)$, as well as the corresponding $\cal C_{\mathrm{FoM}}^\delta$. (Right) compares the performance of these predictive CBFs to the nominal CBF without $\delta$.\vspace{-0.9cm}}
    \label{fig:ds_1d}
\end{figure}

\section{Learning Predictive CBFs}

Predictive CBFs are amiable to implementation, as facilitated through learning.  We first consider the algorithmic computation of PCBFs through simulation roll-outs.  While this effectively approximates PCBFs, the computational overhead prevents their rapid evaluation.  Learning the $\delta$ term circumvents this limitation.  We demonstrate the performance of the predictive CBF and its learned counterpart in simulation in this section, followed by experimentally in the following section.  

\newsec{Computing PCBFs.}  The optimization problem \eqref{eqn:delta_optimization} while difficult to solve exactly, lends itself to a receding horizon implementation. Here, we iteratively approximate solutions to \eqref{eqn:delta_optimization} via simulated rollouts to estimate $\delta(\b x)$ online.  In particular, for $\bx_0 \in \X$, CBF $h$, RoM controller $\bk$, and tracking controller $\bK$, define $\delta^i(\bx_0)$ through the following iterative algorithm: 
\begin{align} \label{eqn:rollout_alg}
    \delta^{i+1}(\bx_0) & := \max\left(0,\delta^{i}(\bx_0) - \eta e_{\delta_i}(\bx_0) \right) \quad  
    & \begin{array}{r}
    \delta^0(\bx_0) = 0, \quad i \in \{ 0, \ldots, \bar{N}\}, \\
    \bx_t' \approx \b x_0 + \int_0^t  \b F_{\mathrm{cl}}(\b x_\tau',\b k_{sf}^{\delta_{i}}(\b  \Pi (\b x_\tau'))) d \tau,  \\
    e_{\delta_i}(\bx_0) = \min_{t \in [0, T]} \dot{h}(\b x_t') + \alpha h(\b \Pi(\b x_t'))
    \end{array}
\end{align}
where $T \in \R_{> 0}$ is the rollout horizon, $\eta \in \R_{> 0}$ is a step size, $\bar{N} \in \mathbb{N}_{\geq 0}$ is the number of iterations, and $\bx_t'$ is an approximated (simulated) solution. 
At each iteration, $i$, we roll out a trajectory under our current value of $\delta = \delta^{i}$. We compute the maximum violation of the barrier condition, $e$, and then take a proportional step on $\delta^{i+1}$ to reduce this violation. This algorithm is iterated until convergence, approximating solutions to \eqref{eqn:delta_optimization}. Online, running to convergence is slow; we can instead take a real-time iterations approach \citep{diehl2005real, diehl2002efficient}, where we take only one step on the optimization problem before updating the system state (and tackling a \emph{new} optimization problem). Assuming safety violation does not change too rapidly, we can achieve convergence to approximately optimal $\delta$ as we progress through time; this may be enough to realize safety practically for complex systems online, as discussed in \cref{sec:hopper}. 

\begin{algorithm}[H]
    \SetAlgoLined
    \SetNlSty{}{}{:} 

    \caption{Predictive CBF Learning}
    \SetKwInOut{Input}{input} \SetKwInOut{Output}{output}
    \label{alg:episodic_learning}
    \Input{CBF $h$, \,RoM controller $\b k$, \,Tracking controller $\b K$, randomly initialized NN parameters $\b \theta_0$}
    \Output{CBF Robustification term $\delta_{\b \theta_{N_{\mathrm{epochs}}}}$}
    \vspace{-0.3cm}
    \BlankLine
    \For{$j = 1, \ldots, N_{\mathrm{epochs}}$}{
        $\;\cal{X}_j \leftarrow \left \{\b x_t^i : \dot{\b x}_t^i = \b F^{\delta_{\b \theta_{j-1}}}_{\mathrm{cl}}(\b x_t^i)\right \}_{i=1}^N$  \hspace{2.1cm}\tcp*[h]{Collect dataset under $\delta_{\b \theta}$}\;\\
        $\cal{E}_j\leftarrow \left\{ e^i_t : e^i_t = \min_{\tau \in [t, T]} \dot{h}(\b x_t^i)+\alpha h(\b \Pi(\b x_t^i))\right \}_{i=1}^N$  \hspace{0.5cm}\tcp*[h]{Compute violation}\;\\
        $\cal D_j \leftarrow \left\{\hat{\delta}_t^i : \hat{\delta}_t^i = \delta_{\b \theta_{j-1}}(\b x_t^i) - \eta_j e_t^i \right\}_{i=1}^N$  \hspace{1.5cm}\tcp*[h]{Compute new targets}\;\\
        $\b \theta_{j} = \argmin_{\b \theta} \cal L(\b \theta, \cal X_j, \cal D_j)$ \hspace{2.2cm}\tcp*[h]{Train NN on new targets}\;
    }
\end{algorithm} 
\vspace{-0.5cm}

\newsec{Learning PCBFs.}  In practice, solving the optimization \eqref{eqn:delta_optimization} requires several roll-outs of the system dynamics.
Computation time for these roll-outs may be a limiting factor when executing this method online; additionally, accuracy of the simulator may introduce some brittleness, as differences between the simulator dynamics and hardware could lead to safety violations. 
We aim to tackle both of these issues by proposing a learning algorithm to approximate $\delta(\b x)$ from simulation data, where we can utilize domain randomization to facilitate transfer to hardware. 

The learning algorithm, outlined in \cref{alg:episodic_learning}, takes in the RoM CBF and nominal controller, and the tracking controller; it trains the CBF robustifcation term $\delta_{\b \theta}$, a neural network with parameters $\b \theta$. For each iteration of the algorithm, the algorithm first collects a set of rollouts, $\cal X_j$, under the current robustification term, according to the dynamics: $\b F^{\delta}_{\mathrm{cl}}(\b x_t) = \b F_{\mathrm{cl}}(\b x_t, \b k_{\mathrm{sf}}^{\delta}(\b \Pi(\b x_t)))$.s
Next, the maximum safety violation (or margin) over a horizon of length $T$ is computed for each point along the trajectory. New targets for $\delta$, denoted $\hat{\delta}$, are computed by subtracting the violation from the current approximation of $\delta$, where $\eta_j$ is a step size parameter. Finally, a new neural network is fit to set of targets $\hat{\delta}$ by minimizing the loss function \eqref{eqn:loss},
where $L_{c, \sigma}$ is the check loss function \citep{koenker1978regression}, parameterized by $\sigma \in (0, 1)$.
Choosing $\sigma > 0.5$ will more heavily penalize cases where $x \leq y$; in our application, this will more heavily penalize underestimates of $\hat{\delta}$, as underestimates may not guarantee safety (while overestimates will). This algorithm is run iteratively, until convergence (or a maximum number of iterations). 

\begin{wrapfigure}[5]{L}{0.46\textwidth} 
  \vspace{-1.25cm} 
  \begin{center}
  \begin{scriptsize}
    \begin{align}\label{eqn:loss}
    \cal L(\b \theta, \cal X, \cal D) &= \frac{1}{|\cal D|}\sum_{\substack{\b x_t^i \in \cal X\\ \hat{\delta}_t^i  \in \cal D}} L_{c, \sigma}\left(\delta_{\b \theta}(\b x_t^i), \mathrm{clip}(\hat{\delta}_t^i, 0, \delta_{\mathrm{max}})\right) \nonumber\\ 
     L_{c, \sigma}(x, y) & = \begin{cases}
        \sigma (y - x) & x \leq y \\ (1 - \sigma)(x - y) & x > y
    \end{cases}
    \end{align}
  \end{scriptsize}
  \end{center}
  \vspace{-.4cm}
  \label{fig:loss}
\end{wrapfigure}
We implement this algorithm in pyTorch \citep{paszke2017automatic} running on an NVIDIA GTX 4080 GPU. Toy systems dynamics are implemented in pyTorch, while more complex robotic systems are simulated using GPU simulation capability in IsaacGym \citep{isaacgym} for massively parallelized data collection.

\begin{example}[Single/Double Integrator]
\emph{
Consider again the setup in Example \ref{example:scalar}.  
In this setting, \cref{fig:ds_1d} highlights the differences between the optimized and learned values of $\delta(\b x)$, as well as a comparison between evolutions of the system under the various safety filters. In the left pane, we see the result of solving the optimization problem \eqref{eqn:delta_optimization}. The gray regions indicate the optimization problem being unsolvable, and the corresponding safe set is outlined in red. Critically, note that the boundaries with the region where $\delta(\b x) = \infty$ are quite sharp; to make the learning problem tractable, we cap the target $\delta$ values at 5. Incapable of learning the sharp boundary of the optimization problem, the learned $\delta_{\b \theta}(\b x)$ is slightly conservative around the boundary (see right pane).
}
\end{example}

\begin{figure}
    \centering
    \includegraphics[width=\textwidth]{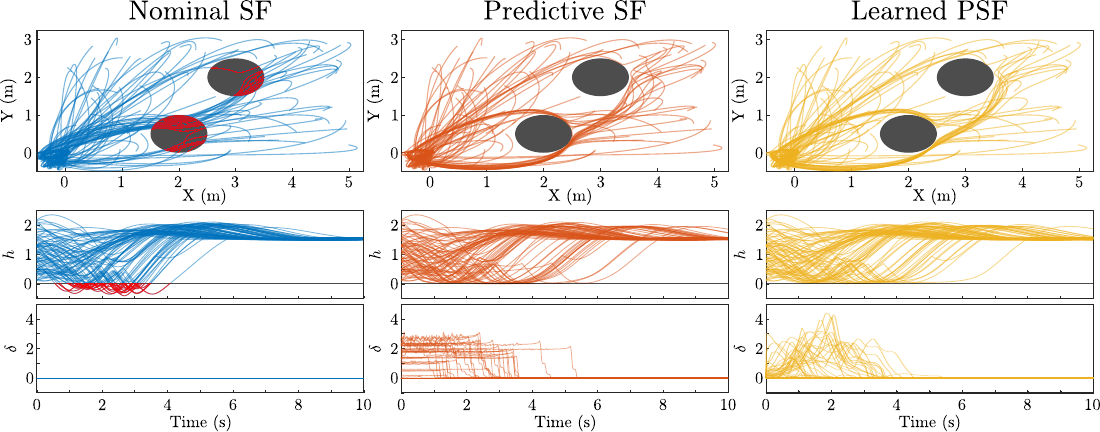}
    \caption{\vspace{-0.9cm}Comparison between the performance of the Predictive CBF (both Learned and Optimized) to the nominal RoM CBF Safety Filter. Both versions of the Predictive CBF ensure safety for the FoM, while the nominal safety filter violates the safety constraint.\vspace{-0.9cm}}
    \label{fig:ds_collision}
\end{figure}

\vspace{-0.6cm}
\begin{example}[Obstacle Avoidance]
\emph{
In a more realistic example, we now take the two-dimensional version of \eqref{eqn:singledoubleint} with $\b z \in \R^2, \b x \in \R^4$, with the goal performing collision avoidance---following from \citep{singletary2021comparative}. The nominal controller is a saturated proportional controller and the nominal barrier function is the distance to the closest obstacle, indexed by $i$:
\begin{equation*}
    \b k(\b z) = -\min\left\{\frac{v_{\mathrm{max}}}{k_p \|\b z\|}, 1\right\}k_p \b z \quad\quad h(\b z) = \min_i\left \{\|\b z - c_i\| - r_i\right\}_{i=1}^{N_{\mathrm{obs}}}
\end{equation*}
Note that this barrier function can be smoothed \citep{glotfelter2017nonsmooth}; we use the non-smooth version as the safe set is more intuitive to characterize and visualize. \cref{fig:ds_collision} displays the difference between the nominal safety filter, optimized PSF, and learned PSF. Because the used tracking controller is slower than the rate the system is allowed to approach the barrier, $\lambda < \alpha$, the assumptions of the nominal safety filter are not satisfied, and the FoM violates the safety constraint. Despite differing $\delta$ profiles, the optimized and learned trajectories are exceptionally similar; the learned version has a slightly smoother $\delta$ profile, but both remain safe over the entire trajectory. 
}
\end{example}

\section{Application to Hardware: 3D Hopping} \label{sec:hopper}
We apply the proposed method to the 3D hopping robot, ARCHER, shown in \cref{fig:hopper}.
The system contains $N=16$ states $\b x = [\b p^\top \; \b q^\top \; \b v^\top \; \b \omega^\top \; \dot{\b \theta}^\top ]^\top$, where $\b p \in \R^3,\b q \in \mathbb{S}^3$ the position and orientation, $\b v \in \R^3, \b \omega \in \R^3$ the linear and angular velocities, and $\dot{\b \theta} \in \R^3$ the flywheels velocities. 
The control input is torque applied to each of the three flywheels.
The foot spring is compressed to a preset distance during the flight phase, and released during the ground phase to maintain constant hop height.
More hardware details can be found \citep{csomay2023nonlinear}. The system dynamics are simulated in IsaacGym \citep{isaacgym} for data generation and learning. 

To accomplish a collision avoidance task, we design a CBF on the 2D single integrator, and track the resulting velocity commands. 
Tracking controllers for ARCHER have been studied in previous works \citep{csomay2024robust, csomay2023nonlinear}. Here we will use the same tracking controller as \citep{compton2025dynamic}, where a Raibert Heuristic \citep{raibert1984experiments} determines desired impact orientations, and these orientations are tracked by a geometrically consistent PD controller.

Using the same nominal controller and barrier function as in the double integrator example, we deploy both real-time iterations of \eqref{eqn:rollout_alg} (Predictive SF) and the Learned PCBF on hardware. To facilitate transfer to hardware (which has worse tracking performance than simulation, primarily due to uncertainty in the position of the center of mass, to which the system is highly sensitive), we domain randomize the robot's center of mass, and use the check loss to upper bound the distribution of resulting behaviors. A comparison with the nominal controller, are shown in \cref{fig:hopper}. Due to unaccounted for tracking errors, the nominal controller violates the safety constraint multiple times on its trajectory to the origin. However, both the predictive SF solved for via optimization and the learned PSF maintain safety over the entire trial. Experiment video is available \citep{video}.

Due to the computationally intensive rollout, the online optimization approach runs at only around 5Hz on hardware, as seen in the bottom right of \cref{fig:hopper}. At this rate, it does not capture how safety prediction varies over the course of a hop, and is reliant upon some native robustness of the CBF condition to enforce safety, as the convergence of $\delta$ to its optimal value is quite slow. On the other hand, the learned algorithm is evaluated at 100Hz on hardware; $\delta$ for the Learned PSF captures variation of safety over the course of a single hop, where delta increases during flight, as the hopper is unable to change its velocity until the next contact. During contact, $\delta$ drops dramatically as the hopper's velocity corrects toward the safe desired velocity. Because of its faster evaluation, the learning approach scales to complex systems and tracking controllers or more dynamic environments, where the computational cost of rollouts is high.

\begin{figure}
    \centering
    \includegraphics[width=\textwidth]{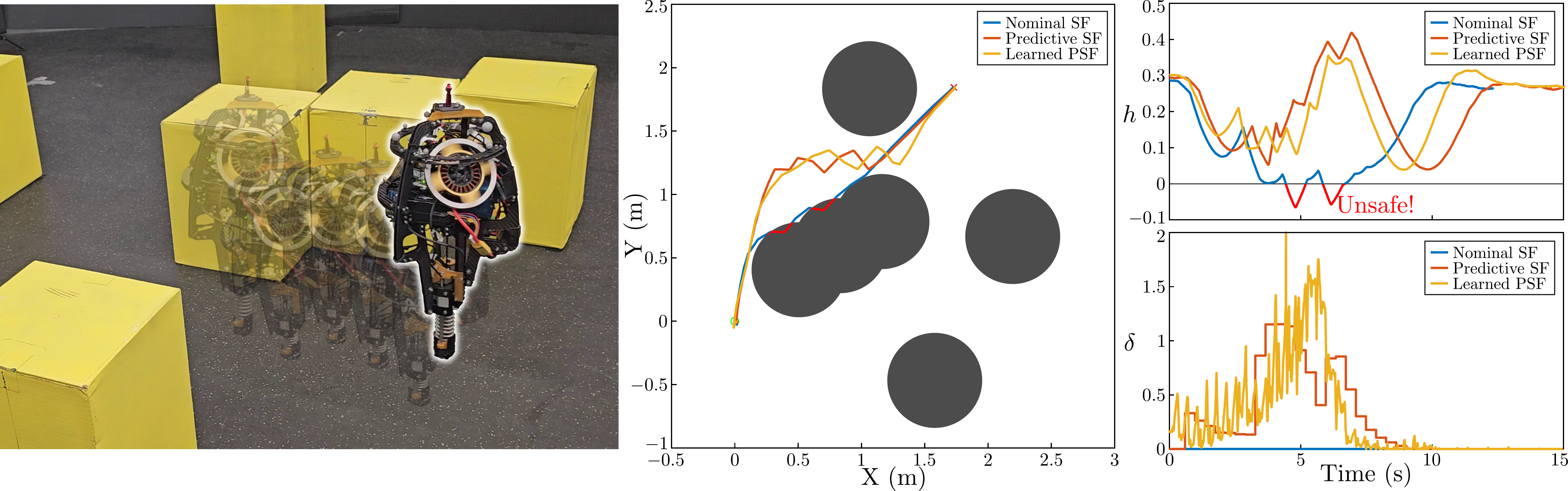}
    \caption{\vspace{-0.9cm}The 3D hopping robot ARCHER (left) navigates a cluttered environment using each of the nominal safety filter (SF), predictive safety filter (PSF) and learned PSF. (Middle) The trajectories of the hopper in space; (Right) values of $h$ and $\delta$ plotted over time for each approach. The nominal SF is unsafe, but both the PSF and LPSF maintain safety.\vspace{-0.9cm}}
    \label{fig:hopper}
\end{figure}

\section{Conclusion}
Synthesis of control barrier functions to ensure safety of complex nonlinear systems is difficult. As safety often depends only on a subset of the states, reduced order models are often used to synthesize safe behaviors, and the behavior of the reduced order model is then tracked on the full order system. We propose a novel, prediction based CBF to account for a very general class of tracking abilities; we robustify the safety filter applied to the RoM by a minimal amount to ensure safety of the FoM over a simulated horizon. As this optimization problem solving for the robustification factor, $\delta(\b x)$, can be prohibitively expensive to run online in real time, we instead learn this term. This method achieves safe navigation through cluttered environments on the 3D hopping robot ARCHER.

\acks{We would like to thank Noel Csomay-Shanklin for his help with experiments.  This research is supported in part by Boeing, and NSF CPS Award \#1932091.}

\bibliography{main}

\begin{thebibliography}{35}
\providecommand{\natexlab}[1]{#1}
\providecommand{\url}[1]{\texttt{#1}}
\expandafter\ifx\csname urlstyle\endcsname\relax
  \providecommand{\doi}[1]{doi: #1}\else
  \providecommand{\doi}{doi: \begingroup \urlstyle{rm}\Url}\fi

\bibitem[Ames et~al.(2014)Ames, Grizzle, and Tabuada]{ames2014control}
Aaron~D Ames, Jessy~W Grizzle, and Paulo Tabuada.
\newblock Control barrier function based quadratic programs with application to adaptive cruise control.
\newblock In \emph{53rd IEEE conference on decision and control}, pages 6271--6278. IEEE, 2014.

\bibitem[Ames et~al.(2016)Ames, Xu, Grizzle, and Tabuada]{ames2016control}
Aaron~D Ames, Xiangru Xu, Jessy~W Grizzle, and Paulo Tabuada.
\newblock Control barrier function based quadratic programs for safety critical systems.
\newblock \emph{IEEE Transactions on Automatic Control}, 62\penalty0 (8):\penalty0 3861--3876, 2016.

\bibitem[Ames et~al.(2019)Ames, Coogan, Egerstedt, Notomista, Sreenath, and Tabuada]{ames2019control}
Aaron~D Ames, Samuel Coogan, Magnus Egerstedt, Gennaro Notomista, Koushil Sreenath, and Paulo Tabuada.
\newblock Control barrier functions: Theory and applications.
\newblock In \emph{2019 18th European control conference (ECC)}, pages 3420--3431. IEEE, 2019.

\bibitem[Breeden and Panagou(2022)]{breeden2022predictive}
Joseph Breeden and Dimitra Panagou.
\newblock Predictive control barrier functions for online safety critical control.
\newblock In \emph{2022 IEEE 61st Conference on Decision and Control (CDC)}, pages 924--931. IEEE, 2022.

\bibitem[Chen et~al.(2021{\natexlab{a}})Chen, Herbert, Hu, Pu, Fisac, Bansal, Han, and Tomlin]{chen2021fastrack}
Mo~Chen, Sylvia~L Herbert, Haimin Hu, Ye~Pu, Jaime~Fernandez Fisac, Somil Bansal, SooJean Han, and Claire~J Tomlin.
\newblock Fastrack: a modular framework for real-time motion planning and guaranteed safe tracking.
\newblock \emph{IEEE Transactions on Automatic Control}, 66\penalty0 (12):\penalty0 5861--5876, 2021{\natexlab{a}}.

\bibitem[Chen et~al.(2021{\natexlab{b}})Chen, Jankovic, Santillo, and Ames]{chen2021backup}
Yuxiao Chen, Mrdjan Jankovic, Mario Santillo, and Aaron~D Ames.
\newblock Backup control barrier functions: Formulation and comparative study.
\newblock In \emph{2021 60th IEEE Conference on Decision and Control (CDC)}, pages 6835--6841. IEEE, 2021{\natexlab{b}}.

\bibitem[Cohen et~al.(2024{\natexlab{a}})Cohen, Csomay-Shanklin, Compton, Molnar, and Ames]{cohen2025safety}
Max~H Cohen, Noel Csomay-Shanklin, William~D Compton, Tamas~G Molnar, and Aaron~D Ames.
\newblock Safety-critical controller synthesis with reduced-order models.
\newblock \emph{arXiv preprint arXiv:2411.16479}, 2024{\natexlab{a}}.

\bibitem[Cohen et~al.(2024{\natexlab{b}})Cohen, Molnar, and Ames]{cohen2024safety}
Max~H Cohen, Tamas~G Molnar, and Aaron~D Ames.
\newblock Safety-critical control for autonomous systems: Control barrier functions via reduced-order models.
\newblock \emph{Annual Reviews in Control}, 57:\penalty0 100947, 2024{\natexlab{b}}.

\bibitem[Compton(2024)]{video}
William~D Compton.
\newblock {Experiment Compilation: Predictive Control Barrier Functions}.
\newblock YouTube, 2024.
\newblock URL \url{https://youtu.be/6pY7T6yucBs}.

\bibitem[Compton et~al.(2024)Compton, Csomay-Shanklin, Johnson, and Ames]{compton2025dynamic}
William~D Compton, Noel Csomay-Shanklin, Cole Johnson, and Aaron~D Ames.
\newblock Dynamic tube mpc: Learning tube dynamics with massively parallel simulation for robust safety in practice.
\newblock \emph{arXiv preprint arXiv:2411.15350}, 2024.

\bibitem[Csomay-Shanklin et~al.(2023)Csomay-Shanklin, Dorobantu, and Ames]{csomay2023nonlinear}
Noel Csomay-Shanklin, Victor~D. Dorobantu, and Aaron~D. Ames.
\newblock Nonlinear model predictive control of a 3d hopping robot: Leveraging lie group integrators for dynamically stable behaviors.
\newblock \emph{Proceedings - IEEE International Conference on Robotics and Automation}, 2023.

\bibitem[Csomay-Shanklin et~al.(2024)Csomay-Shanklin, Compton, Rodriguez, Ambrose, Yue, and Ames]{csomay2024robust}
Noel Csomay-Shanklin, William~D Compton, Ivan Dario~Jimenez Rodriguez, Eric~R Ambrose, Yisong Yue, and Aaron~D Ames.
\newblock Robust agility via learned zero dynamics policies.
\newblock \emph{arXiv preprint arXiv:2409.06125}, 2024.

\bibitem[Diehl et~al.(2002)Diehl, Findeisen, Schwarzkopf, Uslu, Allg{\"o}wer, Bock, Gilles, and Schl{\"o}der]{diehl2002efficient}
Moritz Diehl, Rolf Findeisen, Stefan Schwarzkopf, Ilknur Uslu, Frank Allg{\"o}wer, Hans~Georg Bock, Ernst-Dieter Gilles, and Johannes~P Schl{\"o}der.
\newblock An efficient algorithm for nonlinear model predictive control of large-scale systems part i.
\newblock \emph{Automation technology}, 2002.

\bibitem[Diehl et~al.(2005)Diehl, Bock, and Schl{\"o}der]{diehl2005real}
Moritz Diehl, Hans~Georg Bock, and Johannes~P Schl{\"o}der.
\newblock A real-time iteration scheme for nonlinear optimization in optimal feedback control.
\newblock \emph{SIAM Journal on control and optimization}, 43\penalty0 (5):\penalty0 1714--1736, 2005.

\bibitem[Glotfelter et~al.(2017)Glotfelter, Cort{\'e}s, and Egerstedt]{glotfelter2017nonsmooth}
Paul Glotfelter, Jorge Cort{\'e}s, and Magnus Egerstedt.
\newblock Nonsmooth barrier functions with applications to multi-robot systems.
\newblock \emph{IEEE control systems letters}, 1\penalty0 (2):\penalty0 310--315, 2017.

\bibitem[Koenker and Bassett~Jr(1978)]{koenker1978regression}
Roger Koenker and Gilbert Bassett~Jr.
\newblock Regression quantiles.
\newblock \emph{Econometrica: journal of the Econometric Society}, pages 33--50, 1978.

\bibitem[Langson et~al.(2004)Langson, Chryssochoos, Rakovi{\'c}, and Mayne]{langson2004robust}
Wilbur Langson, Ioannis Chryssochoos, SV~Rakovi{\'c}, and David~Q Mayne.
\newblock Robust model predictive control using tubes.
\newblock \emph{Automatica}, 40\penalty0 (1):\penalty0 125--133, 2004.

\bibitem[Makoviychuk et~al.(2021)Makoviychuk, Wawrzyniak, Guo, Lu, Storey, Macklin, Hoeller, Rudin, Allshire, Handa, and State]{isaacgym}
V~Makoviychuk, L~Wawrzyniak, Y~Guo, M~Lu, K~Storey, M~Macklin, D~Hoeller, N~Rudin, A~Allshire, A~Handa, and G~State.
\newblock Isaac gym: High performance gpu based physics simu-lation for robot learning.
\newblock \emph{Neural Information Processing Systems}, 2021.

\bibitem[Matni et~al.(2024)Matni, Ames, and Doyle]{matni2024quantitative}
Nikolai Matni, Aaron~D Ames, and John~C Doyle.
\newblock A quantitative framework for layered multirate control: Toward a theory of control architecture.
\newblock \emph{IEEE Control Systems Magazine}, 44\penalty0 (3):\penalty0 52--94, 2024.

\bibitem[Mitchell et~al.(2005)Mitchell, Bayen, and Tomlin]{mitchell2005time}
Ian~M Mitchell, Alexandre~M Bayen, and Claire~J Tomlin.
\newblock A time-dependent hamilton-jacobi formulation of reachable sets for continuous dynamic games.
\newblock \emph{IEEE Transactions on automatic control}, 50\penalty0 (7):\penalty0 947--957, 2005.

\bibitem[Molnar et~al.(2022)Molnar, Cosner, Singletary, Ubellacker, and Ames]{Molnar2022}
Tamas~G. Molnar, Ryan~K. Cosner, Andrew~W. Singletary, Wyatt Ubellacker, and Aaron~D. Ames.
\newblock Model-free safety-critical control for robotic systems.
\newblock \emph{IEEE Robotics Automation Letters}, 7\penalty0 (2):\penalty0 944–951, April 2022.
\newblock ISSN 2377-3774.

\bibitem[Nguyen and Sreenath(2016)]{nguyen2016exponential}
Quan Nguyen and Koushil Sreenath.
\newblock Exponential control barrier functions for enforcing high relative-degree safety-critical constraints.
\newblock In \emph{2016 American Control Conference (ACC)}, pages 322--328. IEEE, 2016.

\bibitem[Paszke et~al.(2017)Paszke, Gross, Chintala, Chanan, Yang, DeVito, Lin, Desmaison, Antiga, and Lerer]{paszke2017automatic}
Adam Paszke, Sam Gross, Soumith Chintala, Gregory Chanan, Edward Yang, Zachary DeVito, Zeming Lin, Alban Desmaison, Luca Antiga, and Adam Lerer.
\newblock Automatic differentiation in pytorch.
\newblock \emph{Neural Information Processing}, 2017.

\bibitem[Raibert et~al.(1984)Raibert, Brown~Jr, and Chepponis]{raibert1984experiments}
Marc~H Raibert, H~Benjamin Brown~Jr, and Michael Chepponis.
\newblock Experiments in balance with a 3d one-legged hopping machine.
\newblock \emph{The International Journal of Robotics Research}, 3\penalty0 (2):\penalty0 75--92, 1984.

\bibitem[Robey et~al.(2020)Robey, Hu, Lindemann, Zhang, Dimarogonas, Tu, and Matni]{robey2020learning}
Alexander Robey, Haimin Hu, Lars Lindemann, Hanwen Zhang, Dimos~V Dimarogonas, Stephen Tu, and Nikolai Matni.
\newblock Learning control barrier functions from expert demonstrations.
\newblock In \emph{2020 59th IEEE Conference on Decision and Control (CDC)}, pages 3717--3724. IEEE, 2020.

\bibitem[Schweidel et~al.(2022)Schweidel, Yin, Smith, and Arcak]{schweidel2022safe}
Katherine~S Schweidel, He~Yin, Stanley~W Smith, and Murat Arcak.
\newblock Safe-by-design planner--tracker synthesis with a hierarchy of system models.
\newblock \emph{Annual Reviews in Control}, 53:\penalty0 138--146, 2022.

\bibitem[Singh et~al.(2023)Singh, Landry, Majumdar, Slotine, and Pavone]{singh2023robust}
Sumeet Singh, Benoit Landry, Anirudha Majumdar, Jean-Jacques Slotine, and Marco Pavone.
\newblock Robust feedback motion planning via contraction theory.
\newblock \emph{The International Journal of Robotics Research}, 42\penalty0 (9):\penalty0 655--688, 2023.

\bibitem[Singletary et~al.(2021)Singletary, Klingebiel, Bourne, Browning, Tokumaru, and Ames]{singletary2021comparative}
Andrew Singletary, Karl Klingebiel, Joseph Bourne, Andrew Browning, Phil Tokumaru, and Aaron Ames.
\newblock Comparative analysis of control barrier functions and artificial potential fields for obstacle avoidance.
\newblock In \emph{2021 IEEE/RSJ International Conference on Intelligent Robots and Systems (IROS)}, pages 8129--8136. IEEE, 2021.

\bibitem[Singletary et~al.(2022)Singletary, Guffey, Molnar, Sinnet, and Ames]{singletary2022safety}
Andrew Singletary, William Guffey, Tamas~G Molnar, Ryan Sinnet, and Aaron~D Ames.
\newblock Safety-critical manipulation for collision-free food preparation.
\newblock \emph{IEEE Robotics and Automation Letters}, 7\penalty0 (4):\penalty0 10954--10961, 2022.

\bibitem[Taylor et~al.(2020)Taylor, Singletary, Yue, and Ames]{taylor2020control}
Andrew~J Taylor, Andrew Singletary, Yisong Yue, and Aaron~D Ames.
\newblock A control barrier perspective on episodic learning via projection-to-state safety.
\newblock \emph{IEEE Control Systems Letters}, 5\penalty0 (3):\penalty0 1019--1024, 2020.

\bibitem[Tonkens and Herbert(2022)]{tonkens2022refining}
Sander Tonkens and Sylvia Herbert.
\newblock Refining control barrier functions through hamilton-jacobi reachability.
\newblock In \emph{2022 IEEE/RSJ International Conference on Intelligent Robots and Systems (IROS)}, pages 13355--13362. IEEE, 2022.

\bibitem[Wabersich and Zeilinger(2022)]{wabersich2022predictive}
Kim~P Wabersich and Melanie~N Zeilinger.
\newblock Predictive control barrier functions: Enhanced safety mechanisms for learning-based control.
\newblock \emph{IEEE Transactions on Automatic Control}, 68\penalty0 (5):\penalty0 2638--2651, 2022.

\bibitem[Wabersich et~al.(2023)Wabersich, Taylor, Choi, Sreenath, Tomlin, Ames, and Zeilinger]{wabersich2023data}
Kim~P Wabersich, Andrew~J Taylor, Jason~J Choi, Koushil Sreenath, Claire~J Tomlin, Aaron~D Ames, and Melanie~N Zeilinger.
\newblock Data-driven safety filters: Hamilton-jacobi reachability, control barrier functions, and predictive methods for uncertain systems.
\newblock \emph{IEEE Control Systems Magazine}, 43\penalty0 (5):\penalty0 137--177, 2023.

\bibitem[Wang et~al.(2017)Wang, Ames, and Egerstedt]{wang2017safe}
Li~Wang, Aaron~D Ames, and Magnus Egerstedt.
\newblock Safe certificate-based maneuvers for teams of quadrotors using differential flatness.
\newblock In \emph{2017 IEEE International Conference on Robotics and Automation (ICRA)}, pages 3293--3298. IEEE, 2017.

\bibitem[Xu et~al.(2015)Xu, Tabuada, Grizzle, and Ames]{xu2015robustness}
Xiangru Xu, Paulo Tabuada, Jessy~W Grizzle, and Aaron~D Ames.
\newblock Robustness of control barrier functions for safety critical control.
\newblock \emph{IFAC-PapersOnLine}, 48\penalty0 (27):\penalty0 54--61, 2015.

\end{thebibliography}

\newpage
\appendix
\section{Proof of Theorem \ref{eqn:thmbackground}}
In this section, we provide a proof of Theorem \ref{eqn:thmbackground}, which, although already established in various forms, e.g., in \citep{Molnar2022,cohen2025safety}, is re-derived here to draw parallels and distinctions with the results of the present paper. For convenience, the theorem is restated below:
\begin{untheorem}
Consider a RoM \eqref{eqn:RoM}, a closed-loop FoM \eqref{eqn:FoMcl} and a safe set $\C_{\mathrm{RoM}}$ satisfying Assumptions \ref{asm:projection} and \ref{asm:bounded}. 
For the safety filter $\bv = \b k_{\mathrm{sf}}(\b z)$ in \eqref{eqn:safetyfilter}, if there exists a tracking function $V$ satisfying \eqref{eqn:trackingfunction} with $\lambda \geq \alpha_x$ for any $\alpha_x > \alpha$, then the FoM is safe: 
\begin{eqnarray*}
\bx_0 \in \cal{S} = \left\{ 
\bx \in \X  \colon (\alpha_x - \alpha) h(\bPi(\bx)) - \frac{C_h}{\rho} V(\bx) \geq 0 \right\}
\quad \implies \quad 
\bx_t \in \C_{\mathrm{FoM}}, \:\: \forall t \geq 0. 
\end{eqnarray*}
\end{untheorem}

\begin{proof}
    Under the assumption that $\alpha_x>\alpha$, we note that:
    \begin{eqnarray*}
        \bx\in\mc{S} \implies H(\bx) \geq 0 \implies (\alpha_x - \alpha) h(\bPi(\bx)) \geq \frac{C_h}{\rho} V(\bx) \implies h(\bPi(\bx)) \geq 0 \implies \bx\in\mc{C}_{\mathrm{FoM}},
    \end{eqnarray*}
    so that $\mc{S}\subset\mc{C}_{\mathrm{FoM}}$, where:
    \begin{eqnarray}\label{eqn:H-old-thm}
        H(\bx) \coloneqq (\alpha_x - \alpha) h(\bPi(\bx)) - \frac{C_h}{\rho} V(\bx).
    \end{eqnarray}
    Now, computing the derivative of $h$ along a trajectory of the closed-loop FoM \eqref{eqn:FoMcl} yields:
    \begin{align*}
        \dot{h}(\bPi(\bx_t), \dot{\bx}_t) = & \pdv{h}{\bz}\Big\vert_{\bPi(\bx_t)}\pdv{\bPi}{\bx} \Big \vert_{\bx_t} \bF(\bx_t) + \pdv{h}{\bz}\Big\vert_{\bPi(\bx_t)}\pdv{\bPi}{\bx} \Big \vert_{\bx_t} \bG(\bx_t)\bu \\
        = & \pdv{h}{\bz}\Big\vert_{\bPi(\bx_t)}\bf(\bPi(\bx_t)) + \bg(\bPi(\bx_t))\bPsi(\bx_t) \\ 
        = & \mc{L}_{\bf}h(\bPi(\bx_t)) + \mc{L}_{\bg}h(\bPi(\bx_t))\bPsi(\bx_t),
    \end{align*}
    where the second line follows from Assumption \ref{asm:projection}. Adding and subtracting $ \mc{L}_{\bg}h(\bPi(\bx))\bk_{\mathrm
    sf}(\bPi(\bx))$, where $\bk_{\mathrm
    sf}$ is from \eqref{eqn:safetyfilter}, to the above and lower bounding yields  (suppressing the dependence of time):
    \begin{align*}
        \dot{h}(\bPi(\bx), \dot{\bx}) = & \mc{L}_{\bf}h(\bPi(\bx)) + \mc{L}_{\bg}h(\bPi(\bx))\bk_{\mathrm
    sf}(\bPi(\bx)) + \mc{L}_{\bg}h(\bPi(\bx))\left[\bPsi(\bx) - \bk_{\mathrm
    sf}(\bPi(\bx)) \right] \\
    \geq & -\alpha h(\bPi(\bx)) - \|\mc{L}_{\bg}(\bPi(\bx))\|\cdot\|\bPsi(\bx) - \bk_{\mathrm
    sf}(\bPi(\bx))\| \\
    \geq & -\alpha h(\bPi(\bx)) - \frac
    {C_h}{\rho} V(\bx),
    \end{align*}
    where the final inequality follows from Assumption \ref{asm:bounded} and \eqref{eqn:trackingfunction}. Using the above, we may bound the derivative of $H$ along the closed-loop FoM as:
    \begin{align}
        \dot{H}(\bx) = & (\alpha_x - \alpha)\dot{h}(\bPi(\bx), \dot{\bx}) - \frac{C_h}{\rho}\dot{V}(\bx) \\
        \geq & -\alpha(\alpha_x - \alpha)h(\bPi(\bx)) -\alpha_x \frac
    {C_h}{\rho} V(\bx) + \alpha \frac
    {C_h}{\rho} V(\bx) + \lambda \frac
    {C_h}{\rho} V(\bx) \\
    = & -\alpha H(\bx) + \frac{C_h}{\rho}\left[\lambda - \alpha_x \right] V(\bx),\label{eqn:H-dot-bound-old-thm}
    \end{align}
    where the final equality follows from substituting in $H$ via \eqref{eqn:H-old-thm}. Provided that $\lambda \geq \alpha_x$, then $\dot{H}(\bx)\geq - \alpha H(\bx)$, which implies that $H$ is a barrier function \citep{ames2016control} for the closed-loop FoM \eqref{eqn:FoMcl} and that $\mc{S}$ is forward invariant. As $\mc{S}$ is forward invariant and $\mc{S}\subset\mc{C}_{\mathrm{FoM}}$, any initial condition $\bx_0\in\mc{S}$ ensures that $\bx_t\in\mc{C}_{\mathrm{FoM}}$ for all $t\geq0$, as desired.
\end{proof}

\section{Proof of Theorem \ref{thm:predcbfs-exists}}
In this section, we provide a proof of Theorem \ref{thm:predcbfs-exists}, restated here for convenience:
\begin{untheorem}
Consider a RoM \eqref{eqn:RoM}, a closed-loop FoM \eqref{eqn:FoMcl} and a safe set $\C_{\mathrm{RoM}}$ satisfying Assumptions \ref{asm:projection} and \ref{asm:bounded}.  Let $\delta_0(\bx) \equiv \delta_0$ for $\delta_0 \in \R_{\geq 0}$. If there exists a $\delta_0$-bounded tracking function, $V_{\delta_0}$, with $\lambda \geq \alpha_x$ for any $\alpha_x > \alpha$, then the FoM is safe: 
\begin{equation}
\begin{array}{rcl}
\delta_0 & \geq & \frac{C_h \mu}{\alpha_x \rho}\\
\bx_0 \in \cal{S}_{\delta_0} & : = & \C_{\mathrm{FoM}} \cap \left\{ 
\bx \in \X  \colon (\alpha_x - \alpha) h(\bPi(\bx)) + \delta_0 - \frac{C_h}{\rho} V_{\delta_0}(\bx) \geq 0 \right\}
\end{array}
\:\: \implies  \:\: 
\begin{array}{r}
\bx_t \in \C_{\mathrm{FoM}} \\
\:\: \forall t \geq 0. 
\end{array} \nonumber
\end{equation}
Moreover, the set $\cal{S}_{\delta_0} \subset \cal C_{\mathrm{FoM}}$ and the constant function $\delta_0:\cal{S}_{\delta_0}  \to \R_{\geq 0}$ satisfies $\delta(\b x) \leq \delta_0$ for all $\b x \in \cal{S}_{\delta_0}$.  Therefore, $\delta : \C_{\mathrm{FoM}}^{\delta} \to \R_{\geq 0}$ exists with $\cal{S}_{\delta_0}  \subset \C_{\mathrm{FoM}}^{\delta}$. 
\end{untheorem}

\begin{proof}
    Similar to the proof of Theorem \ref{eqn:thmbackground}, under Assumption \ref{asm:projection}, the derivative of $h$ along the closed-loop FoM \eqref{eqn:FoMcl} with $\bv=\bk_{\mathrm{sf}}^{\delta_0}(\bPi(\bx), \bx)$ is:
    \begin{align}
        \dot{h}(\bPi(\bx_t)), \dot{\bx}_t) = \mc{L}_{\bf}h(\bPi(\bx_t)) + \mc{L}_{\bg}h(\bPi(\bx_t))\bPsi(\bx_t).
    \end{align}
    Adding and subtracting $ \mc{L}_{\bg}h(\bPi(\bx))\bk_{\mathrm
    sf}^{\delta_0}(\bPi(\bx), \bx)$, where $\bk^{\delta_0}_{\mathrm
    sf}$ is from \eqref{eqn:predictive_safetyfilter}, to the above and lower bounding yields (suppressing the dependence of time):
    \begin{align}
        \dot{h}(\bPi(\bx), \dot{\bx}) = & \mc{L}_{\bf}h(\bPi(\bx)) + \mc{L}_{\bg}h(\bPi(\bx))\bk^{\delta_0}_{\mathrm
    sf}(\bPi(\bx), \bx) + \mc{L}_{\bg}h(\bPi(\bx))\left[\bPsi(\bx) - \bk^{\delta_0}_{\mathrm
    sf}(\bPi(\bx),\bx) \right] \nonumber \\
    \geq & -\alpha h(\bPi(\bx)) + \delta_0 - \|\mc{L}_{\bg}(\bPi(\bx))\|\cdot\|\bPsi(\bx) - \bk^{\delta_0}_{\mathrm
    sf}(\bPi(\bx),\bx)\| \nonumber \\
    \geq & -\alpha h(\bPi(\bx)) + \delta_0 - \frac
    {C_h}{\rho} V_{\delta_0}(\bx), \label{eqn:hdot-bound-new-thm}
    \end{align}
    where the final inequality follows from Assumption \ref{asm:bounded} and \eqref{eqn:trackingfunction}. Now, consider the function:
    \begin{equation}
        H(\bx) = \dot{h}(\bPi(\bx)), \dot{\bx}) + \alpha_x h(\bPi(\bx)) = \mc{L}_{\bf}h(\bPi(\bx)) + \mc{L}_{\bg}h(\bPi(\bx))\bPsi(\bx) + \alpha_x h(\bPi(\bx)),
    \end{equation}
    which defines a set:
    \begin{equation}
        \mc{S} \coloneqq \left\{\bx\in\mc{X}\,:\,H(\bx) \geq 0 \right\}.
    \end{equation}
    While $\mc{S}$ is not necessarily a subset of $\mc{C}_{\mathrm{FoM}}$, we note that if $\bx_t\in\mc{S}$ for all $t\geq0$, then:
    \begin{equation*}
        \dot{h}(\bPi(\bx_t)), \dot{\bx}_t) = \mc{L}_{\bf}h(\bPi(\bx_t)) + \mc{L}_{\bg}h(\bPi(\bx_t))\bPsi(\bx_t) \geq - \alpha_x h(\bPi(\bx_t)),
    \end{equation*}
    for all $t\geq 0$. Hence, if $\bx_t\in\mc{S}$ for all $t\geq0$ and $\bx_0\in\mc{S}\cap\mc{C}_{\mathrm{FoM}}$, it follows from \cite{ames2016control} that $h(\bPi(\bx_t))\geq0$ for all $t\geq0$. To this end, note that:
    \begin{align}
        H(\bx) = &  \dot{h}(\bPi(\bx)), \dot{\bx}) + \alpha_x h(\bPi(\bx)) \\
        \geq & -\alpha h(\bPi(\bx)) + \delta_0 - \frac
    {C_h}{\rho} V_{\delta_0}(\bx) + \alpha_x h(\bPi(\bx)),
    \end{align}
    where the inequality follows from \eqref{eqn:hdot-bound-new-thm}. Defining:
    \begin{align}
        H_{\delta_0}(\bx) \coloneqq & (\alpha_x - \alpha) h(\bPi(\bx)) + \delta_0 - \frac
    {C_h}{\rho} V_{\delta_0}(\bx) \label{eqn:H-delta} \\
    \mc{G}_{\delta_0} = & \left\{\bx\in\mc{X}\,:\, H_{\delta_0}(\bx)\geq 0\right\},
    \end{align}
    and using the observation that $H(\bx)\geq H_{\delta_0}(\bx)$, we have that $\mc{G}_{\delta_0}\subset\mc{S}$ as $H_{\delta_0}(\bx)\geq0$ implies $H(\bx)\geq0$. Hence, if $\mc{G}_{\delta_0}$ is forward invariant, then $\bx_0\in\mc{S}_{\delta_0}=\mc{G}_{\delta_0
    }\cap\mc{C}_{\mathrm{FoM}}$ implies that $h(\bPi(\bx_t))\geq0$ for all $t\geq0$. To show this, we compute:
    \begin{align}
        \dot{H}_{\delta_0}(\bx) = & (\alpha_x - \alpha)\dot{h}(\bPi(\bx), \dot{\bx}) - \frac{C_h}{\rho}\dot{V}(\bx) \nonumber \\
        \geq & -\alpha(\alpha_x - \alpha)h(\bPi(\bx)) + \alpha_x\delta_0 - \alpha\delta_0 -\alpha_x \frac
        {C_h}{\rho} V(\bx) + \alpha \frac
        {C_h}{\rho} V(\bx) + \lambda \frac
        {C_h}{\rho} V(\bx) - \frac{C_h\mu}{\rho} \nonumber \\
    = & -\alpha\left[(\alpha_x - \alpha)h(\bPi(\bx)) + \delta_0 - \frac
        {C_h}{\rho} V(\bx) \right] + \alpha_x \delta_0  -\alpha_x \frac
        {C_h}{\rho} V(\bx) + \lambda \frac
        {C_h}{\rho} V(\bx) - \frac{C_h\mu}{\rho} \nonumber \\
        = & -\alpha H_{\delta_0}(\bx) + \frac{C_h}{\rho}\left[\lambda - \alpha_x \right] V(\bx) + \alpha_x\delta_0 - \frac{C_h\mu}{\rho},
    \end{align}
    where the first inequality follows from \eqref{eqn:hdot-bound-new-thm} and \eqref{eqn:deltatrackingfunction} and the final equality follows from \eqref{eqn:H-delta}. Now, provided that:
    \begin{equation*}
        \lambda \geq \alpha_x,\quad\quad \delta_0 \geq \frac{C_h\mu}{\alpha_x\rho},
    \end{equation*}
    then $\dot{H}_{\delta_0}(\bx)\geq -\alpha H_{\delta_0}(\bx)$ so that $H_{\delta_0}$ is a barrier function for the closed-loop FoM \eqref{eqn:FoMcl}, implying that $\mc{G}_{\delta_0}$ is forward invariant \citep{ames2016control}. As $\mc{G}_{\delta_0}$ is forward invariant and $\mc{G}_{\delta_0}\subset\mc{S}$ it follows that $\bx_0\in\mc{S}_{\delta_0}$ implies that $\bx_t\in\mc{C}_{\mathrm{FoM}}$ for all $t\geq0$, thereby ensuring safety of the FoM.

    To show that $\delta(\bx)$ exists, we will show it is well defined on $\mc{S}_{\delta_0}$ and can be bounded by $\delta_0$.   To this end, let $\bx \in \mc{S}_{\delta_0}$.  Then the optimization problem \eqref{eqn:delta_optimization} can be stated using the notation of this proof, and bounded according to: 
    \begin{eqnarray}
    \delta(\bx) & =  & 
    \left\{\begin{array}{rlr}
    \min _{\delta \in \R_{\geq 0}} &  \quad \delta  & \\
\mathrm{s.t.}    &  
\quad H(\b x_t) \geq 0,    &   \forall t \geq 0  \\
 &  
 \quad 
\L_f h(\b z)  + \L_g h(\b z)  \b v^{\delta}(\b z) \geq -\alpha h(\b z) + \delta,   & \qquad \b z = \Pi(\b x)  \\
 & 
 \quad \dot{\b  x}_t = \b  F_{\mathrm{cl}}(\b x_t, \b K(\b x_t,\b v^{\delta}(\b \Pi(\b x_t))),  &  \b x_0 = \b x .  
 \end{array} \right. \nonumber\\
 & \leq & 
 \left\{\begin{array}{rlr}
    \min _{\delta \in \R_{\geq 0}} &  \quad \delta  & \\
\mathrm{s.t.}    &  
\quad \delta \geq \frac{C_h\mu}{\alpha_x\rho} \\
& \quad  H_{\delta}(\b x_t) \geq 0,    &   \forall t \geq 0  \\
 &  
 \quad 
\L_f h(\b z)  + \L_g h(\b z) \bk^{\delta}_{\mathrm
    sf}(\bPi(\bx), \bx) \geq -\alpha h(\b z) + \delta,   & \qquad \b z = \Pi(\b x)  \\
 & 
 \quad \dot{\b  x}_t = \b  F_{\mathrm{cl}}(\b x_t, \b K(\b x_t,\bk^{\delta}_{\mathrm
    sf}(\bPi(\bx), \bx)),  &  \b x_0 = \b x .  
 \end{array} \right. \nonumber\\
 & = & 
 \left\{\begin{array}{rlr}
    \min _{\delta \in \R_{\geq 0}} &  \quad \delta  & \\
\mathrm{s.t.}    &  
\quad \delta \geq \frac{C_h\mu}{\alpha_x\rho}   
 \end{array} \right. \nonumber\\
 & = & \frac{C_h\mu}{\alpha_x\rho} \nonumber\\
 & \leq & \delta_0 \nonumber
    \end{eqnarray}
Therefore, $\delta(\bx) \in [0,\delta_0]$ for all $\bx \in \mc{S}_{\delta_0}$.  As a result, it has a feasible solution and is well defined for all $\bx \in \mc{S}_{\delta_0}$.  This also implies that $\mc{S}_{\delta_0}\subseteq\mc{C}^{\delta}_{\mathrm{FoM}}$ as desired.


\end{proof}

\end{document}